\documentclass{article}
\usepackage[utf8]{inputenc}
\usepackage{amsmath}
\usepackage{amssymb}
\usepackage{amsthm}
\usepackage{algpseudocode}
\usepackage{hyperref}
\theoremstyle{plain}

\usepackage{url}
\usepackage{float}
\usepackage{comment}

\DeclareMathOperator{\lca}{\text{lca}}

\newtheorem{thm}{Theorem}[section]
\newtheorem{claim}[thm]{Claim}
\newtheorem{lemma}[thm]{Lemma}

\theoremstyle{definition}
\newtheorem{definition}[thm]{Definition}

\theoremstyle{remark}

\title{Computing Truncated Metric Dimension of Trees}
\author{Paul Gutkovich, Zi Song Yeoh}
\begin{document}

\maketitle

\begin{abstract}
Let $G=(V,E)$ be a simple, unweighted, connected graph. Let $d(u,v)$ denote the distance between vertices $u,v$. A resolving set of $G$ is a subset $S$ of $V$ such that knowing the distance from a vertex $v$ to every vertex in $S$ uniquely identifies $v$. The metric dimension of $G$ is defined as the size of the smallest resolving set of $G$. We define the $k$-truncated resolving set and $k$-truncated metric dimension of a graph similarly, but with the notion of distance replaced with $d_k(u,v) := \min(d(u,v),k+1)$. 

In this paper, we demonstrate that computing $k$-truncated dimension of trees is NP-Hard for general $k$. We then present a polynomial-time algorithm to compute $k$-truncated dimension of trees when $k$ is a fixed constant.
\end{abstract}

\section{Introduction}
For any set of three non-collinear points in the Euclidean plane, any point in the plane can be determined by its distances to the points in the set. In general, in $d$-dimensional space any set of $d+1$ linearly independent points allows us to determine any point by only knowing distances to those points.

If the Euclidean space is replaced by graphs, the problem becomes more challenging. Our task is to find the smallest set of vertices in a graph such that any point can be uniquely determined based on its distances to the vertices in the set. The size of this smallest set is called the graph's metric dimension. For a positive integer $k$, we define $k$-truncated metric dimension, which is a similar concept, with the exception that the usual distance metric is replaced by the $k$-truncated distance metric. For two vertices $u,v$ in a graph, their $k$-truncated distance, denoted $d_k(u,v)$, is defined as $\min(d(u,v), k+1)$.

A lot of previous work has been done on analyzing metric dimension (see \cite{bib:dim-survey} for a survey on current results). In this paper, our main focus will be on algorithms in computing truncated metric dimension of trees. 

Computing metric dimension of trees is known to be easy (doable in linear-time); see \cite{bib:tree-dim-2000,bib:set-cover-dim-upper,bib:slater,bib:hararymelter} for example. In contrast, less is known about algorithms for $k$-truncated metric dimension. It is also known that computing the $k$-truncated metric dimension for a general graph is NP-complete \cite{bib:dimk-np}. For trees, there exists a polynomial-time algorithm to compute the $1$-truncated metric dimension (also known as adjacency dimension) of a tree based on dynamic programming \cite{bib:1-truncated-tree}. We show that for general values of $k$, there is no polynomial-time algorithm for computing $k$-truncated metric dimension of trees unless $\mathbf{P}=\mathbf{NP}$. In contrast, we show that for any constant $k$, there is a polynomial-time algorithm for computing $k$-truncated metric dimension of trees.

In Section~\ref{sec:basic}, we provide basic definitions and notations that will be used throughout the paper. In Section~\ref{sec:np-hardness}, we prove that computing $k$-truncated metric dimension is NP-hard for general values of $k$. In Section~\ref{sec:poly-alg}, we show that if $k$ is a fixed constant, there is a polynomial-time algorithm for computing $k$-truncated dimension of trees.

\section{Preliminary Definitions and Notations}\label{sec:basic}
Throughout this paper, we let $G=(V,E)$ be a simple, undirected, connected graph, and $n=|V|$. For any two vertices $u,v\in V$, define the graph distance $d(u,v)$ to be the number of edges in the shortest path between $u$ and $v$. For a positive integer $k$, we define the $k$-truncated distance $d_k(u,v)=\min(d(u,v), k+1)$. We make some preliminary definitions that will be used throughout the paper.

\begin{definition}[$k$-close, $k$-far, and $k$-dominated]
For two vertices $u,v \in V$, we say that $u$ is \emph{$k$-close} to $v$ if $d(u,v)\leq k$, and $u$ is \emph{$k$-far} from $v$ otherwise. We say a vertex $u$ is \emph{$k$-dominated} by a set $S \subseteq V$ if it is $k$-close to some vertex $v \in S$.
\end{definition}

\begin{definition}[$k$-distinguishing vertex]\label{def:distinguishing_vertex}
We say a vertex $w$ \emph{$k$-distinguishes} vertices $u$ and $v$ if $d_k(u,w) \neq d_k(v,w)$. We call $w$ a \emph{$k$-distinguishing vertex} of $(u,v)$.
\end{definition}

\begin{definition}[$k$-truncated (dominating) resolving set]\label{def:resolving_set}
A subset $S\subseteq V$ is called a \emph{$k$-truncated resolving set} of $G$ if for any $u,v \in V$, there exists some vertex $x \in S$ that $k$-distinguishes $u$ and $v$. We also define a $k$-truncated dominating resolving set of $G$ to be a set $S'$ that satisfies the following conditions:
\begin{itemize}
    \item $S'$ is a $k$-truncated resolving set of $T$.
    \item Every element of $V$ is $k$-dominated by $S'$.
\end{itemize}
\end{definition}
\begin{definition}[$k$-truncated metric dimension]\label{def:metric_dimension}
We define the $k$-truncated metric dimension of $G$ as the size of the smallest $k$-truncated resolving set of $G$. The $k$-truncated metric dimension of a graph $G$ is denoted as $\dim_k(G)$. We define the minimum $k$-truncated dominating resolving set of $G$ to be the smallest $k$-truncated dominating resolving set of $G$. 
\end{definition}

\section{NP-Hardness of Computing Truncated Metric Dimension on Trees}\label{sec:np-hardness}
Given that computing metric dimension of general graphs is hard, it is no surprise that computing $k$-truncated metric dimensions in general graphs is also hard. It has been shown that computing the $k$-truncated metric dimension for a general graph is NP-hard \cite{bib:dimk-np}. 

In this section, we will show that computing $k$-truncated metric dimension of general trees is NP-hard for general values of $k$.
\begin{thm}\label{thm:np-hardness}
Unless $\mathbf{P}=\mathbf{NP}$, for any constant $c>0$, there is no algorithm that computes $k$-truncated metric dimension on trees in $O(n^c)$ time for all $k \le n$.
\end{thm}

Throughout this section, let $f_k(T)$ denote the size of the smallest $k$-truncated dominating resolving set of $T$. For technical reasons, it would be easier to deal with $f_k(T)$. The two quantities $f_k(T)$ and $\dim_k(T)$ (the truncated metric dimension of $T$) differ by at most $1$.
\begin{lemma}\label{lem:dim_mld_close}
For any tree $T$ and positive integer $k$, we have $\dim_k(T)\leq f_k(T)\leq \dim_k(T)+1$.
\end{lemma}
\begin{proof}
Since every $k$-truncated dominating resolving set of $T$ is automatically a $k$-resolving set of $T$, we get $\dim_k(T)\leq f_k(T)$. Let $R$ be a minimal $k$-truncated resolving set of $T$. We have two cases.
\begin{itemize}
    \item \textbf{Case 1: There exists a vertex $a\in V(T)$ such that $d(a,x)>k$ for all $x\in R$.}
    
    In this case, at most one such vertex $a$ can exist, because if two such vertices $a$ and $b$ exist, then $d_k(a,x)=d_k(b,x)=k+1$ for all $x\in R$, contradicting the definition of $R$. 
    Thus, if we define $M=R\cup \{a\}$, then $M$ will be a $k$-truncated dominating resolving set. Hence, $f_k(T)\leq |M|=\dim_k(T)+1$.
    \item \textbf{Case 2: For all vertices $a\in V(T)$, there exists some $x\in R$ such that $d(a,x)\leq k$.}\\ In this case, $R$ satisfies both conditions for a $k$-truncated dominating resolving set, so $f_k(T)\leq |R|=\dim_k(T)$.
\end{itemize}
In both cases, we get $f_k(T)\leq \dim_k(T)+1$.
\end{proof}
In order to prove NP-hardness, we reduce our problem to $3$-dimensional matching, which is known to be NP-hard.
\begin{definition}[$3$-Dimensional Matching]
Let $m$ be a positive integer and let $X,Y,Z$ be pairwise disjoint sets of integers with size $m$. Let $U$ be a subset of $X\times Y\times Z$. Every element in $U$ is of the form $(x_i,y_i,z_i)$, where $x_i\in X$, $y_i\in Y$, $z_i\in Z$. The $3$-dimensional matching problem is the problem of computing the largest $W \subseteq U$ such that for any $(x_i,y_i,z_i),(x_j,y_j,z_j)\in W$, we have $x_i\neq x_j, y_i\neq y_j, z_i\neq z_j$. 
\end{definition}
We will require the following NP-hardness result.
\begin{thm}\label{thm:3dm-nphard}\cite{bib:3dm-approx}
Approximating $3$-dimensional matching to a multiplicative factor of $\frac{95}{94}$ is NP-hard.
\end{thm}
We now prove Theorem~\ref{thm:np-hardness}.
\begin{proof}[Proof of Theorem~\ref{thm:np-hardness}]
We will prove this by reduction from the $3$-dimensional matching problem. Fix an instance of the $3$-dimensional matching problem, i.e. let $m$ be a positive integer, and let $X,Y,Z$ be pairwise disjoint subsets of $\{1,2, \ldots, m\}$ with size $n$. Let $U$ be a subset of $X\times Y\times Z$. Let $M$ be the maximum number of disjoint elements of $U$.

Let $r=|U|$, and let $S_i = \{x_i,y_i,z_i\}$ be the $i^{th}$ element of $U$ for $i\in \{1,2,\ldots, r\}$. We may assume without loss of generality $r>2$ and $r-M\geq 2$ (the latter assumption can be made by appending an independent instance whose answer is at least $2$ less than the number of triples). Let $k=2m+2$. Let $T$ be a tree with root $u$, with children $u_1, u_2, \ldots, u_r$. Let $T_i$ be the subtree rooted at $u_i$. For every $i\in \{1, 2, \ldots, r\}$, construct $T_i$ as follows:

\begin{itemize}
    \item Let $w_{i,1}, w_{i,2}, \ldots w_{i,2k}$ be vertices such that $w_{i,1}$ is the root of $T_i$ (i.e. $w_{i,1}=u_i$) and $w_{i,j}$ is connected to $w_{i,j+1}$ for all $1 \le j \le 2k-1$. Let $a_{i,1}, a_{i,2}, a_{i,3}$ be the distinct elements of $S_i$. Create vertices 
    $x_{i,1,1},$ $x_{i,1,2}, \ldots, x_{i,1,a_{i,1}+1}$ such that $x_{i,1,1}$ is a child of $w_{i,a_{i,1}}$ and $x_{i,1,j+1}$ is the child of $x_{i,1,j}$ for all $1 \le j \le a_{i,1}$. Repeat this process for $a_{i,2}, a_{i,3}$. For brevity, define $g_{i,c}:=x_{i,c,a_{i,c}+1}$ for $c\in \{1,2,3\}$.
\end{itemize}

After $T$ is constructed, let $R$ be a smallest $k$-truncated dominating resolving set of $T$. For $i\in\{1, 2, \ldots, r\}$, define $p_i:=|T_i\cap R|$. Note that $p_i\geq 1$ for all $i$, because there must be an element of $R$ within distance $k$ of $w_{i, 2k}$, implying there must be an element of $R$ in $T_i$.
\begin{claim}
We have $|R|\geq 2r-M$.
\end{claim}
\begin{proof}
Assume for contradiction that $|R|<2r-M$, i.e. $2r-|R|>M$. First, note that $\sum_i p_i = |R|$ if $u\not\in R$ and $\sum_i p_i = |R|-1$ if $u\in R$. Because $p_i\geq 1$ for all $i$, there must be at least $2r-|R|$ values of $i$ such that $p_i=1$. From our assumption, more than $M$ values of $i$ satisfy $p_i=1$. Without loss of generality, assume that $p_1=p_2=\cdots=p_{M+1}=1$. 

Consider some subtree $T_i$, where $i\in\{1, 2, \ldots, M+1\}$. It has only one marked vertex $v_i$ (a vertex that is in $R$), which must be within distance $k$ of $w_{i,2k}$. This implies that $v_i\in \{w_{i,k}, w_{i,k+1}, \ldots, w_{i, 2k}\}$. For any possible $v_i$, and for any $c\in\{1,2,3\}$, we get $d(v_i, g_{i,c})\geq d(w_{i,k}, g_{i,c})=d(w_{i,k},w_{i,a_{i,c}})+d(w_{i,a_{i,c}}, g_{i,c})=(k-a_{i,c})+(a_{i,c}+1)=k+1$. This implies that $v_i$ is more than $k$ away from any $g_{j,c}$ for any $j\in \{1, 2, \ldots, r\}$ and any $c\in \{1, 2, 3\}$. This implies that for any $i\in \{1, 2, \ldots, M+1\}$ and any $c\in \{1, 2, 3\}$, there must be some vertex $y\not\in T_1\cup T_2\cup \cdots \cup T_{M+1}$ such that $y$ is within $k$ of $g_{i,c}$. Note that to find this $y$, we only need to consider one vertex, which is the vertex closest to $u$ in $\{u\}\cup T_{M+2}\cup T_{M+3}\cup \cdots \cup T_r$. From now on, $y$ will denote this vertex. Thus, we have that if $g_{i_1,c_1}$ and $g_{i_2, c_2}$, with $i_1,i_2 \in \{1, 2, \ldots, M+1\}, c_1, c_2\in\{1, 2, 3\}$ are $k$-distinguished by some vertex in $R$ (i.e. the $k$-truncated distances from them to the vertex are different), then they are $k$-distinguished by $y$. We show that $y$ cannot $k$-distinguish all pairs $g_{i_1,c_1}$ and $g_{i_2, c_2}$.

Note that if $y$ $k$-distinguishes some $g_{i_1,c_1}$ and $g_{i_2, c_2}$, with $i_1,i_2 \in \{1, 2, \ldots, M+1\}, c_1, c_2\in\{1, 2, 3\}$, then $u$ $k$-distinguishes them as well. So, we may assume $y=u$.

Note that for some $i\in\{1, 2, \ldots, M+1\}, c\in\{1, 2, 3\}$, we have $d(u,g_{i,c})=d(u,w_{i,a_{i,c}})+d(w_{i,a_{i,c}}, x_{i,c,a_{i,c}+1})=a_{i,c}+a_{i,c}+1=2a_{i,c}+1$. If $u$ $k$-distinguishes all pairs $g_{i_1,c_1}, g_{i_2, c_2}$, with $i_1,i_2 \in \{1, 2, \ldots, M+1\}, c_1, c_2\in\{1, 2, 3\}$, then $\min(k+1, 2a_{i,c}+1)$ is distinct for all $i \in \{1, 2, \ldots, M+1\}, c \in\{1, 2, 3\}$. Since $k=2m+2>2a_{i,c}$, all $a_{i,c}$ are distinct for all $i \in \{1, 2, \ldots, M+1\}, c \in\{1, 2, 3\}$. However, this is equivalent to saying that $S_1, S_2, \ldots, S_{M+1}$ are all pairwise disjoint, which contradicts the maximality of $M$. This gives a contradiction, which implies $|R| \ge 2r-M$.
\end{proof}

\begin{claim}
There exists a $k$-truncated dominating resolving set $R$ of $T$ with size $2r-M$.
\end{claim}
\begin{proof}
Assume that the largest non-overlapping subset of $U$ is $S_1, S_2, \ldots, S_M$. Construct $R$ by including $b_i:=w_{i,k}$ for all $i\in\{1, 2, \ldots, r\}$ (call these bottom vertices) and $t_i:=u_i$ for all $i\in \{M+1, M+2, \ldots, r\}$ (call these top vertices). We show that this set $R$ works. 

First, we show that all vertices are within distance $k$ of some vertex in $R$. Note that $u$ is close to $t_r$. Let $v$ be any vertex in $T_i$. The only cases where $v$ is not close to $b_i$ are if $v=g_{i,c}$ for some $c\in\{1,2,3\}$. However, in this case we have $d(t_r, g_{i,c})\leq 2a_{i,c}+2\leq 2m+2=k$, meaning $v$ is close to $t_r$.

Now, we will show that all vertices of $T$ are $k$-distinguished by some vertex in $R$. For this, let $q$ be some unknown vertex in $T$. Suppose all we know is the value of $d_k(q,y)$ for all $y\in R$, and we want to uniquely determine $q$. We have a few cases:
\begin{itemize}
    \item \textbf{Case 1: $q$ is not $k$-close to any top vertex.}
    
    In this case, $q$ must be within distance $k$ of exactly one bottom vertex, $b_i$. The only possible value of $q$ is $w_{i, k+d_k(q, b_i)}$.
    \item \textbf{Case 2: $q$ is $k$-close to some top vertex.}
    
    If $q$ is $k$-close to only one top vertex, then let that top vertex be $t_i$. Since we assumed $r-M \ge 2$, the only options for $q$ are $b_i=w_{i,k}$ and $w_{i,k+1}$, and $q$ can therefore be determined by its distance to $b_i$.
    
    Now assume $q$ is $k$-close to at least two top vertices. In this case, $q$ must be $k$-close to all top vertices. 
    
    If $q$ is equidistant to all top vertices, then it must be in $T_1\cup T_2\cup \cdots \cup T_M\cup \{u\}$. If $d(q,t_r)=1$, then $q=u$. Assume $d(q,t_r)>1$. 
    
    If $q$ is also $k$-close to some bottom vertex $b_i$, then it is in $T_i$. Let $w'$ be the closest vertex of the form $w_{i,j}$ to $q$. We know that $d(q,w')=d(q,t_r)+d(q,b_i)-k-1$, and $d(u,w')=d(u,q)-d(q,w')=d(q,t_r)-1-d(q,w')$. From these two quantities we can find out what $q$ is.
    
    If $q$ is not $k$-close to any bottom vertices, then it must be of the form $g_{i,c}$ for some $i\in\{1, 2, \ldots, M\}$, and $d(t_r, q)=2a_{i,c}+2 \le k$. Because none of the $S_i$ overlap for $i\in\{1, 2, \ldots, M\}$, $q$ can be determined solely by its distance from $t_r$.
    
    If $q$ is not equidistant to all top vertices, then there is some $i\in\{M+1, M+2, \ldots, r\}$ such that $d(q,t_i)$, is minimized. This means $q$ is in $T_i$. Assume that $q$ is $k$-close to $b_i$. Then, let $w'$ be the closest vertex to $q$ of the form $w_{i,j}$ for some $j$. Then $d(w', q)=d(t_i, q)+d(b_i, q)-(k-1)$, and $d(u, w') = d(u, q) - d(q, w') = d(q,t_r) + 1 - d(q,w')$. From these two quantities we can determine $q$.

    Finally, if $q$ is not $k$-close to $b_i$, then $q=g_{i,c}$ for some $c\in\{1, 2, 3\}$. Thus, $q$ can be determined from $d(q, t_i)$.
\end{itemize}

This shows that $R$ is a $k$-truncated dominating resolving set of $T$. Thus, there exists a $k$-truncated dominating resolving set of size $2r-M$.
\end{proof}

Together with Lemma~\ref{lem:dim_mld_close}, we see that the $k$-truncated metric dimension of $T$ is in the range $[2r-M-1,2r-M]$.

Now, assume that we can compute the minimum $k$-resolving set of $T$ in $O(n^c)$  time. By Lemma \ref{lem:dim_mld_close}, this means that we can approximate the minimum truncated dominating resolving set to within an additive factor of $1$, which means that we approximated the 3-Dimensional Matching problem for $U$ within an additive factor of $1$. However, this contradicts Theorem~\ref{thm:3dm-nphard}. Thus, computing $k$-truncated metric dimension on trees is NP-Hard, as desired.
\end{proof}

\section{Polynomial-Time Algorithm to Compute \\$k$-truncated Metric Dimension for Constant $k$}\label{sec:poly-alg}

There exists a polynomial-time algorithm to compute the $1$-truncated metric dimension (also known as adjacency dimension) of a rooted tree based on dynamic programming \cite{bib:1-truncated-tree}. We generalize this result to any arbitrary constant $k$.

Fix a positive integer $k$. In this section, we present an algorithm to compute the $k$-truncated metric dimension of a tree that runs in time polynomial in $n$, the number of vertices of the tree. For the sake of simplicity, we present a slightly unoptimized version of our algorithm, which runs in $O(n^2)$. It is possible to improve our algorithm to $O(n)$, but this is not the main focus on the paper.

First, we show an algorithm that computes the smallest $k$-truncated dominating resolving set of a tree $T$. We then show how this algorithm can be slightly modified to compute the $k$-truncated metric dimension of $T$.
    
\subsection{Computing the smallest $k$-truncated dominating resolving set}
We present a dynamic programming algorithm to compute the size of the smallest $k$-truncated dominating resolving set for any fixed constant $k$.

Let $T$ be a rooted tree with $n$ vertices and let $k$ be a positive integer. Label the vertices in the tree with integers $1, 2, \ldots, n$. For any vertex $u$ of $T$, let $\text{ch}(u)$ denote the number of children $u$ has, and let $T_u$ denote the subtree of $T$ rooted at $u$. For a vertex $u$ and positive integer $m\leq \text{ch}(u)$, define $c_m(u)$ to be the child of $u$ with the $m^{th}$ smallest label. For any $u\in V(T)$ and positive integer $m\leq \text{ch}(u)$, we define $T_{u,m}=\{u\}\cup T_{c_1(u)}\cup T_{c_2(u)}\cup \cdots \cup T_{c_m(u)}$. We define a function $f(u, C, D, E, l, m)$ as the size of the smallest set $S$ satisfying the following properties:
\begin{itemize}
    \item $S$ is a subset of $T'$, where $T':=T_{u,m}$.
    \item If $v$ is any vertex outside of $T'$ with $d(u,v)=l$, then $S\cup\{v\}$ is a $k$-truncated dominating resolving set of $T'$. If $l$ is $null$, then $S$ is a $k$-truncated dominating resolving set of $T'$. Note that if $l>k$, it cannot $k$-distinguish or be $k$-close to elements of $T'$, so we will only allow $l$ to be an integer between $1$ and $k$ or $null$.
    \item The set $\{d(u,z)\;|\;z\in S,\; d(u,z)\leq2k\}$ is equal to $C$.
    \item The set $\{d(u,y)\;|\;y\in T',\;d(u,y)\leq 2k:\;\forall z\in S,\;d(y,z)>k\}$ is equal to $D$.
    \item Call an $x \in V(T') \setminus S$ \textit{good} if it satisfies the following conditions:
     \begin{itemize}
        \item $d(x,u)\leq 2k$.
        \item there exists $z \in S$ where $d(x,z) \le k$.
        \item $d(x,z)-d(u,z)$ is equal for all $z\in S$ with $d(x,z) \le k$.
    \end{itemize}

    We start with an empty set $E$. For each \textit{good} $x$, we add the tuple $(d(x,u), d(x,z)-d(u,z), d(u,z'))$ to $E$, where $z \in S$ and $d(x,z) \le k$, and $z' \in S$ is a vertex with minimal $d(u,z')$ under the condition $d(x,z') > k$. If such $z'$ does not exist, let $d(u,z')=null$.
\end{itemize}
We define $g(u,C,D,E,l):=f(u,C,D,E,l,\text{ch}(u))$. Note that the size of the smallest $k$-truncated dominating resolving set of $T$ is just the minimum of $g(root,C,D,E,null)$ over all valid $C,D,E$. 

We will now show a recursive algorithm to compute $f$. Originally, set $f(u,C,D,E,l,m)$ to $null$ for all valid tuples of $u,C,D,E,l,m$. The rest of the proof is a long casework to describe all the necessary state transitions.

\paragraph{Base Case.} If $u$ is a leaf, the only possible values for $f(u,C,D,E,l,m)$ are $0, 1$, corresponding to the sets $\emptyset, \{u\}$. Because $u$ has no children, we only allow $m=0$. The valid combinations of $C,D,E,l$ that give $f(u,C,D,E,l,m)=0$ are $C=\emptyset, D=\{0\}$, $E=\emptyset$, $l\neq null$. The valid combinations of $C,D,E,l$ that give $f(u,C,D,E,l,m)=1$ are $C=\{0\}$, $D=\emptyset$, $E=\emptyset$, $l=null$.

\paragraph{Recursive Step: Adding a Parent.} From now on, we assume that $u$ is not a leaf. Suppose $m=1$. Let $v=c_1(u)$ and let $C, D, E, l$ be such that $g(v,C,D,l)$ is not $null$. Let $S$ be the corresponding set of marked vertices. Let $l'$ be the distance from $u$ to the closest marked vertex outside $\{u\}\cup T_v$. 

We have the following observations:
\begin{itemize}
    \item If $l=1$, then $u$ is marked, and $l'$ can take any value.
    \item If $2\leq l \leq k$, then $l'=l-1$.
    \item If $l=null$, then $l'$ is either $null$ or $k$.
\end{itemize} 
All we need to check is if $S$, along with the vertex corresponding to $l'$, $k$-locates and $k$-dominates $u$. If $l=1$, then $u$ is forced to be marked, giving $f(u,C',D',E',l',1)=|S|+1$. Assume instead that $1< l'\leq k$. We know that $u$ is $k$-dominated by $S$ because $l'\leq k$. Also, $u$ is the closest vertex to the vertex corresponding to $l'$, so it is automatically $k$-distinguished from all vertices of $T_v$, giving $f(u,C',D',E',l',1)=|S|$, where $C',D',E'$ are the values of the $C,D,E$ parameters corresponding to $T_v\cup\{u\}$ and set $S$. Now assume that $l=l'=null$. For any $x\in V(T_v)$, we want to check if $S$ $k$-distinguishes $u$ and $x$.

If there exist two distinct $z_1,z_2\in S$ such that $z_1,z_2$ are $k$-close to $x$ and $d(x,z_1)-d(v,z_1)\neq d(x,z_2)-d(v,z_2)$, then we claim that one of $z_1, z_2$ distinguishes $(x,u)$. To see this, note that
\begin{align*}
d(x,z_1)=d(u,z_1)
&\implies d(x,z_1)-d(v,z_1)=1
\\&\implies d(x,z_2)-d(v,z_2)\neq 1
\\&\implies d(x,z_2)\neq d(u,z_2),
\end{align*}
which implies that if $z_1$ does not distinguish $x,u$, then $z_2$ does. 

Now, suppose such $z_1,z_2$ do not exist. Note that $l'=null$ implies $x$ must be \textit{good}, i.e. $(d(x,v),d(x,z)-d(v,z), d(v,z'))\in E$ for some $z,z' \in S$ (where $z'$ might not exist). If $x,u$ are distinguished by some $z\in S$ that is $k$-close to $x$, then we have $d(x,z)\neq d(u,z)\iff d(x,z)-d(v,z)\neq 1$. If $x,u$ are distinguished by some other $z'\in S$ then it is necessary and sufficient for $d(v,z')\leq k-1$. Thus, $x,u$ are distinguished by $S$ if and only if $d(x,z)-d(v,z)\neq 1$ or $d(v,z')\leq k-1$. By checking if this is true for all elements of $E$, we can determine whether $S$ is a $k$-truncated dominating resolving set of $T_{v}\cup\{u\}$. If all $x,u$ pairs are distinguished, then $f(u,C',D',E',l',1)=|S|$. 

\paragraph{Recursive Step: Merging a child subtree.} Now assume $m>1$. Let $v=c_m(u)$, $T'=T_{u,m}$, $T_1=T_{u,m-1}$ and $T_2=T_v$. Pick valid $C_1, D_1, E_1, l_1$ such that $f(u,C_1,D_1,E_1,l_1,m-1)$ is not $null$ (call the set it corresponds to $S_1$), and valid $C_2, D_2, E_2, l_2$ such that $g(v,C_2,D_2,E_2,l_2)$ is not $null$ (call the set it corresponds to $S_2$). Let $u_1,u_2$ be vertices outside of $T_1,T_2$ respectively such that $d(u,u_1)=l_1$, $d(v,u_2)=l_2$. Firstly, we want to check if $S_1,S_2,l_1,l_2$ are compatible. We have a few cases.
\begin{itemize}
\item \textbf{Case 1: $u_1 \in S_2$ and $u_2 \in S_1$.}\\
For this to be true, we need to check that $l_1-1=\min(C_2)$ and $l_2-1=\min(C_1)$. In this case, we can let $l'$, the distance to the closest marked vertex to $u$ outside of $T'$, to be any integer between $\max(l_1, l_2-1)$ and $k+1$, or $null$.

\item \textbf{Case 2: $u_1 \in S_2$ and $u_2 \not\in T'$.}\\
For this, we need to check that $l_1-1=\min(C_2)$, that $l_2\leq 1+\min(C_1)$, and that $l_1\leq l_2-1$. In this case, $l'=l_2-1$.

\item \textbf{Case 3: $u_2 \in S_1$ and $u_1 \not\in T'$.}\\
For this, we need to check that $l_2-1=\min(C_1)$, $l_1\leq 1+\min(C_2)$, and $l_2\leq l_1+1$. In this case, $l'=l_1$.

\item \textbf{Case 4: $u_1,u_2 \not\in T'$.}\\
Here, we must have $u_1=u_2$. For this case, we need to check that $l_1+1=l_2$, $l_1\leq 1+\min(C_2)$, and $l_2\leq 1+\min(C_1)$. In this case, $l'=l_1$.

\item \textbf{Case 5: $l_1=null$ and $l_2 \neq null$.}\\
We require that $\min(C_2)+1>k$. Either $u_2 \in S_1$, which requires $l_2=\min(C_1)+1$ and allows for $l'$ to be any integer between $l_2-1$ and $k$ or $null$, or $u_2 \not\in T$, which requires $l_2\leq \min(C_1)+1$ and gives $l'=l_2-1$.

\item \textbf{Case 6: $l_1 \neq null$ and $l_2=null$.}\\
We require that $\min(C_1)+1>k$. Either $u_1 \in S_2$, which requires $l_1=\min(C_2)+1$ and allows for $l'$ to be any integer between $l_1$ and $k$ or $null$, or $u_1 \not\in T'$, which requires $l_1\leq \min(C_2)+1$ and gives $l'=l_1$.

\item \textbf{Case 7: $l_1=l_2=null$.}\\
In this case, we require $\min(C_1)+1, \min(C_2)+1>k$. We must have $l'=null$.
\end{itemize}

Only when one of the above cases is true, we can proceed with the $l'$ determined by that case. Now, we want to check if $S':=S_1\cup S_2 \cup \{u_1,u_2\}$ forms a $k$-truncated dominating resolving set of $T_{u,m}$. It suffices to check that for all $x\in V(T_1), y\in V(T_2)$, there is some $z$ in $S'$ that $k$-distinguishes $x$ and $y$. We have a few cases.
\begin{itemize}
\item \textbf{Case 1: some $z_1\in S_1$ is $k$-close to $x$, and some $z_2\in S_2$ is $k$-close to $y$}\\
In this case, we claim that $x$ and $y$ are $k$-distinguished by one of $z_1,z_2$. Assume for contradiction this is not true. Then
$d(x,z_1)=d(y,z_1), d(x,z_2)=d(y,z_2)$. Let $a=\lca(x,z_1),b=\lca(y,z_2)$. Then
\begin{align*}
0
&=d(x,z_1)-d(y,z_1)
\\&=d(x,a)-d(y,a)
\\&=(d(x,z_2)-d(a,z_2))-d(y,a)
\\&=d(y,z_2)-d(a,z_2)-d(y,a)
\\&=d(y,b)-d(a,b)-d(y,a)
\\&=d(y,b)-d(a,b)-(d(y,b)+d(a,b))
\\&=-2d(a,b)<0,
\end{align*}
a contradiction, as desired.

\item \textbf{Case 2: $x$ is not $k$-close to any element of $S_1$ and $y$ is not $k$-close to any element of $S_2$}\\
First, we need to check if either $x$ is $k$-close to some element of $S_2$ or $y$ is $k$-close to some element of $S_1$. Thus, we must check if either $d(x,u)+\min(C_2)+1$ or $d(y,v)+\min(C_1)+1$ are less than $k+1$. Note that $d(x,u)\in D_1, d(y,v)\in D_2$. If neither of these conditions are true, then we need to check if $x$ and $y$ are distinguished by a vertex outside $T'$; i.e. they are distinguished by the vertex corresponding to $l'$. For this, we need to check that $\min(k+1, d(x,u)+l')\neq \min(k+1, d(y,v)+l'+1)$. If this condition is not satisfied, then $S'$ does not $k$-distinguish $(x,y)$. Otherwise, $S'$ does distinguish $(x,y)$.

\item \textbf{Case 3: $x$ is $k$-close to some element of $S_1$, $y$ is not $k$-close to any element of $S_2$}

Note that $d(v,y)\in D_2$. If there exist two $z_1, z_2\in S_1$ that are $k$-close to $x$ such that $d(x,z_1)-d(u,z_1)\neq d(x,z_2)-d(u,z_2)$, then we claim that either $z_1$ or $z_2$ $k$-distinguished $x$ and $y$. To show this, assume that $z_1$ does not $k$-distinguish $x$ and $y$. Then
\begin{align*}
d(x,z_1)
=d(y,z_1)
&\implies d(x,z_1)-d(u,z_1)=d(u,y)
\\&\implies d(x,z_2)-d(u,z_2)\neq d(u,y)
\\&\implies d(x,z_2)\neq d(z_2,y),
\end{align*}
which means $z_2$ distinguishes $x,y$.

From now on, assume that such $z_1,z_2$ do not exist. This means that $d(x,z)-d(u,z)$ is constant for all $z\in S_1$ that are $k$-close to $x$. This means that $E_1$ contains the tuple $(d(u,x), d(x,z)-d(u,z), d(u,z'))$, where $z$ is an arbitrary element of $S_1$ that is $k$-close to $x$, while $z'$ is the closest vertex to $u$ of all elements of $S_1$ that are not $k$-close to $x$. We will now check if $x,y$ are distinguished. We have a few cases:
\begin{itemize}
    \item \textbf{Subcase 3.1: $x,y$ are distinguished by an element of $S_2$}\\ We just need to check if $x$ is $k$-close to an element of $S_2$. For this, we just need to see if $d(x,u)+\min(C_2)+1\leq k$.
    \item \textbf{Subcase 3.2: $x,y$ are distinguished by an element of $S_1$} \\For any $z\in S_1$ that is $k$-close to $x$, we have $d(x, z)\neq d(y,z) \iff d(x,z)-d(u,z)\neq d(u,y)$, so it suffices to check that $d(x,z)-d(u,z)\neq d(u,y)=1+d(v,y)$. For any $z' \in S_1$ not $k$-close to $x$, we need to check that the closest such $z'$ to $u$ is $k$-close to $y$, i.e. we need to ensure that $d(u,z')+d(v,y)+1\leq k$.
    \item \textbf{Subcase 3.3: $x,y$ are distinguished by a marked vertex outside of $T'$, i.e. by the vertex corresponding to $l'$}\\
    For $x,y$ to be distinguished, we require $d(x,u)+l'\neq d(y,v)+l'+1$, which is equivalent to $ d(x,u)\neq d(y,v)+1$ and $\min(d(x,u)+l', d(y,v)+l'+1)\leq k$.
\end{itemize}
Vertices $x$ and $y$ are distinguished if and only if at least one of these cases is true. 

\item \textbf{Case 4: $y$ is $k$-close to some element of $S_2$, $x$ is not $k$-close to any element of $S_1$}\\
This case is analogous to Case 3.

We must ensure that, for any $x\in V(T_1), y\in V(T_2)$, one of the above cases is true. We will now show an algorithm that does all of these checks in $O(A(k))$ time, where $A$ is some function that has only $k$ as a variable (i.e. is independent of $n$). The algorithm will consist of the following steps:
\begin{itemize}
    \item Check that for every $d_1\in D_1, d_2\in D_2$, we have $\min(k+1, d_1+l')\neq \min(k+1, d_2+l'+1)$. This checks all scenarios of Case 2.
    \item For every tuple $(d(u,x), d(x,z)-d(u,z), d(u,z'))\in E_1$ and every $d_2\in D_2$, check that at least one of the following is true:
    \begin{itemize}
        \item See if $d(x,u)+\min(C_2)+1\leq k$. This checks Case 3 Subcase 1.
        \item See if $d(x,z)-d(u,z)\neq 1+d_2$ or $d(u,z')+d_2+1\leq k$. This checks Case 3 Subcase 2.
        \item See if $d(x,u)\neq d_2+1$ and $\min(d(x,u)+l', d_2+l'+1)\leq k$. This checks Case 3 Subcase 3.
    \end{itemize}
    This checks all scenarios of Case 3.
    \item For every tuple $(d(v,y), d(y,z)-d(v,z), d(v,z'))\in E_2$ and every $d_1\in D_1$, check that at least one of the following is true:
        \begin{itemize}
        \item See if $d(y,v)+\min(C_1)+1\leq k$. This checks Case 4 Subcase 1.
        \item See if $d(y,z)-d(v,z)\neq 1+d_1$ or $d(v,z')+d_1+1\leq k$. This checks Case 4 Subcase 2.
        \item See if $d(y,v)\neq d_1-1$ and $\min(d(y,v)+l'+1, d_1+l')\leq k$. This checks Case 4 Subcase 3.
    \end{itemize}
    This checks all scenarios of Case 3.
\end{itemize}
\end{itemize}
Note that the time to complete each step depends only on $k$ and not on $n$, and by completing the steps we check all scenarios of the cases described above. If during the algorithm, one of the scenarios does not satisfy the conditions, we know that $S'$, along with the vertex corresponding to $l'$ is not a valid $k$-truncated dominating resolving set of $T'$. If all scenarios pass their conditions, then we get a valid $k$-truncated dominating resolving set for $T'$. 

Now, we will show how to get parameters $C,D,E$ for $T',S'$ from parameters $C_1,D_1,E_1,C_2,D_2,E_2$. We have that \[C=C_1\cup \{c+1\;|\;c\in C_2, c\leq 2k-1\}\]
\[D=\{d\;|\;d\in D_1, d+\min(C_2)+1>k\}\]
\[\cup \;\{d+1\;|\;d\in D_2, d\leq 2k-1, d+\min(C_1)+1>k\}.\]
Consider $x,z_1,z'$ such that $(d(x,u), d(x,z_1)-d(u,z_1),d(u,z'))\in E_1$. If $x$ is not $k$-close to any elements of $S_2$ (i.e. $d(u,x)+\min(C_2)+1>k$), then $E$ should include $(d(x,u), d(x,z_1)-d(u,z_1),\min(d(u,z'), \min(C_2)+1))$. Assume instead that $x$ is $k$-close to $z_2\in S_2$. Let $c$ be the smallest element of $C_2$ such that $d(u,x)+1+c>k$. We have $d(x,z_2)-d(u,z_2)=d(u,x)$. Thus $E$ must include $(d(x,u), d(x,z_1)-d(u,z_1),\min(d(u,z'), c+1))$ if and only if $d(u,x)\neq d(x,z_1)-d(u,z_1)$.

Similar reasoning holds for tuples in $E_2$. Let $(d(y,v),d(y,z_1)-d(v,z_1),d(v,z'))$ be an arbitrary element of $E_2$. If $y$ is not $k$-close to any element of $C_1$ (i.e. $d(y,v)+1+\min(C_1)>k$), then $E$ should include \\$(d(y,v), d(y,z_1)-d(v,z_1),\min(d(v,z'), \min(C_1)))$. Assume instead that $y$ is $k$-close to some $z_2\in S_1$. Let $c$ be the smallest element of $C_1$ such that $d(u,x)+1+c>k$. We have $d(y,z_2)-d(v,z_2)=d(y,v)$. Thus, $E$ includes $(d(y,v), d(y,z_1)-d(v,z_1), \min(d(y,z'), c))$ if and only if $d(y,z_1)-d(v,z_1)=d(y,v)$. Thus, we can get $C,D,E$ from $C_1,D_1,E_1,C_2,D_2,E_2$ in time that depends only on $k$.

If none of the scenarios in the algorithm fail their checks, and if we have that $f(u,C,D,E,l,m)$ is $null$ or larger than $|S'|$, then we set it to $|S'|$. By running through all valid combinations of $C_1, D_1, E_1, l_1, C_2, D_2, E_2, l_2$, and all valid $l'$ (the number of which is a function of only $k$), we will find the smallest possible value of $f(u,C,D,E,l,m)$, as desired.

Note that because each recursive step has runtime depending only on $n$, computing the size of the minimal $k$-truncated dominating resolving set of $T$, which is just the minimum of $g(root, C, D, E, null)$ over all valid $C,D,E$, has runtime that is linear in $n$ for fixed $k$.
\subsection{Computing $k$-truncated metric dimension}
We will now show how to modify the algorithm in the previous section to compute $\dim_k(T)$, the $k$-truncated metric dimension of $T$. Note that the only difference is that now a vertex is allowed to be $k$-far from all vertices in the resolving set.

Let $T$ be a tree and $k$ be a positive integer. Let $S$ be the smallest $k$-resolving set of $T$. We know that either every element of $V(T)$ is $k$-close to some element of $S$, or exactly one element of $V(T)$ is not $k$-close to any element of $S$. In the first case, $S$ is the minimum $k$-truncated dominating resolving set of $T$, meaning it can be found with the algorithm in the previous section. Let the size of the minimal $k$-truncated dominating resolving set of $T$ be $s_{\text{mdr}}$. Now let $r$ be an arbitrary element of $V(T)$, and assume that $r$ is $k$-far from every element of $S$. Root the tree $T$ at $r$. We will now demonstrate how the algorithm from Section~\ref{sec:poly-alg} can be modified to give the minimal set $S$ that resolves $T$, while every element of $S$ is $k$-far from $r$.

Let $v_1, v_2, \ldots, v_p$ be the children of $u$. We compute $g(v_i, C_i, D_i, E_i, null)$ for all $i\in \{1, 2, \ldots, p\}$ and all valid $C_i,D_i,E_i$, using the algorithm described in the previous section. We want to ensure that $u$ is $k$-far from all marked vertices, which is why we require $l=null$, which implies that $D_i$ must be $\emptyset$. This means we also require $\min(C_i)\geq k$.

We now want to compute the smallest resolving set $S$ of $T$, such that the only vertex not $k$-dominated by $S$ in $T$ is $r$. Let $S_i$ be $S\cap T_i(r)$ for $i\in\{1, 2, \ldots, p\}$. We know that for all $i$, $S_i$ must be the set corresponding to $g(v_i, C_i, D_i=\emptyset, E_i, l_i=null)$, for some valid $C_i$, $E_i$ satisfying $\min(C_i)>k$. Note that because $l_i=null$, $S_i$ on its own must be a $k$-truncated dominating resolving set of $T_i(r)$. Let us define $q_i$ to be the smallest value of $g(v_i, C_i, D_i=\emptyset, E_i, l_i=null)$ over all valid $C_i, E_i$ satisfying $\min(C_i)\geq k$. Then we must have that $|S|=q_1+q_2+\cdots+q_p$. 

Let $s_{\min}$ be the minimal value of $|S|$ over all $r\in V(T)$. We must have that $\dim_k(T)=\min(s_{\text{mdr}}, s_{\min})$. This finishes the description of the algorithm.

\section{Conclusion and Future Work}

In this paper, we focused on computing the truncated metric dimension of trees. We showed that computing $k$-truncated metric dimension of trees is NP-hard for general $k$, but for any constant $k$ it can be solved in polynomial time.

Many open questions remain regarding the computation of $k$-truncated metric dimension.
\begin{itemize}
    \item What is the best dependence on $k$ we can get in an algorithm to compute $\dim_k(T)$, the $k$-truncated metric dimension of a tree $T$?
    \item It is known that for general graphs, the best approximation ratio for computing metric dimension is $\Theta(\log n)$ \cite{bib:lower-bound-dim}. On the other hand, it is possible to compute metric dimension of trees in linear time \cite{bib:tree-dim-2000}. What is the best approximation ratio we can obtain for $k$-truncated metric dimension of trees?
    \item It is known that we cannot efficiently compute $k$-truncated metric dimension in general graphs (even when $k$ is a small constant) \cite{bib:lower-bound-dim}. However, can we efficiently compute truncated metric dimension in other classes of graphs for any constant $k$?
\end{itemize}   

\section{Acknowledgements}

We thank Jesse Geneson for providing the project and discussions about problems related to truncated metric dimension. We also thank Felix Gotti and Tanya Khovanova for reviewing the paper and providing many helpful comments to improve its presentation. Finally, we are grateful for the MIT PRIMES program for the opportunity to carry out this research.

\end{document}